\documentclass[11pt]{article}
\usepackage{graphicx}
\usepackage{subfigure}
\usepackage{tikz}
\usepackage{tikz-network}
\usepackage{float}
\usetikzlibrary{positioning,petri}
\usepackage{amsfonts}
\usepackage{amsbsy}
\usepackage{amssymb}
\usepackage{amsmath,amsthm}
\usepackage{verbatim}
\usepackage{amsfonts}
\usepackage{parskip}
\usepackage[T1]{fontenc}
\usepackage{multicol}
\usepackage[onehalfspacing]{setspace}
\usepackage{color}
\usepackage[autostyle]{csquotes}
\usepackage{pgf}
\usepackage{mathtools}
\usepackage{hyperref}
\usepackage{thmtools}
\usepackage{caption}
 
\declaretheoremstyle[%
  spaceabove=-6pt,%
  spacebelow=6pt,%
  headfont=\normalfont\itshape,%
  postheadspace=1em,%
  qed=\qedsymbol%
]{mystyle} 

\hypersetup{
    colorlinks=true,
    linkcolor=blue,
    filecolor=blue,      
    urlcolor=blue,
    citecolor=blue,
    pdftitle={Overleaf Example},
    pdfpagemode=FullScreen,
    }

\newtheorem{thm}{Theorem}[section]
\newtheorem{prop}{Proposition}[section]
\newtheorem{cor}[thm]{Corollary}
\newtheorem{defn}[thm]{Definition}

\newtheorem{lemma}[thm]{Lemma}

\newtheorem{remark}[prop]{Remark}
\newtheorem{example}[thm]{Example}

\newcommand{\FF}{{\rm I\!F}}
\newcommand{\CC}{\mathcal{C}}
\newcommand{\CD}{\mathcal{D}}
\newcommand{\CE}{\mathcal{E}}
\newcommand{\CL}{\mathcal{L}}
\newcommand{\CO}{\mathcal{O}}
\newcommand{\CQ}{\mathcal{Q}}
\newcommand{\CX}{\mathcal{X}}
\newcommand{\ba}{\mathbf{a}}
\newcommand{\bx}{\mathbf{x}}
\newcommand{\by}{\mathbf{y}}
\newcommand{\lmd}{\mathrm{lmd}}
\newcommand{\spt}{\mathtt{supp}}

\title{On Linear Complementary Pairs of Algebraic Geometry Codes over Finite Fields}

\author{Sanjit Bhowmick$^{1}$, Deepak Kumar Dalai$^{1}$, and Sihem Mesnager$^{2}$\footnote{
$^{1}$School of Mathematical Sciences,
National Institute of Science Education and Research,\\
An OCC of Homi Bhabha National Institute, Bhubaneswar, Odisha 752050, India. Email: sanjitbhowmick@niser.ac.in; deepak@niser.ac.in\\
$^{2}$Department of Mathematics, University of Paris VIII, F-93526 Saint-Denis, Laboratory Analysis, Geometry and Applications, LAGA, University Sorbonne Paris Nord, CNRS, UMR 7539, F-93430, Villetaneuse, France and Telecom Paris, Polytechnic institute of Paris, 91120 Palaiseau, France. Email: smesnager@univ-paris8.fr 
}
}

\begin{document}
\maketitle
\begin{abstract} 

Linear complementary dual (LCD) codes and linear complementary pairs (LCP) of codes have been proposed for new applications as countermeasures against side-channel attacks (SCA) and fault injection attacks (FIA) in the context of direct sum masking (DSM). The countermeasure against FIA may lead to a vulnerability for SCA when the whole algorithm needs to be masked (in environments like smart cards). This led to a variant of the LCD and LCP problems, where several results have been obtained intensively for LCD  codes, but only partial results have been derived for LCP codes. Given the gap between the thin results and their particular importance, this paper aims to reduce this by further studying the LCP of codes in special code families and,  precisely, the characterisation and construction mechanism of LCP codes of algebraic geometry codes over finite fields. Notably, we propose constructing explicit LCP of codes from elliptic curves. Besides, we also study the security parameters of the derived LCP of codes $(\CC, \CD)$ (notably for cyclic codes), which are given by the minimum distances $d(\CC)$ and $d(\CD^\perp)$.
Further, we show that for  LCP  algebraic geometry codes $(\CC,\CD)$, the dual code $\CC^\perp$ is equivalent to $\CD$ under some specific conditions we exhibit. Finally, we investigate whether MDS LCP of algebraic geometry codes exist (MDS codes are among the most important in coding theory due to their theoretical significance and practical interests). Construction schemes for obtaining LCD codes from any algebraic curve were given in 2018 by Mesnager, Tang and Qi in [``Complementary dual algebraic geometry codes", IEEE Trans. Inform Theory, vol. 64(4),  2390--3297, 2018]. To our knowledge, it is the first time LCP of algebraic geometry codes has been studied.

\end{abstract}

 \noindent\textbf{2020 AMS Classification Code:} 51E22; 94B05. 




\section{Introduction}\label{sec:intr}

Let $\CC$ be a linear code of length $n$,  dimension $k$  and minimum Hamming distance $d$ over a  finite field $\mathbb F_q$, where $q = p^m$ and $p$ is a prime. The code $\CC$ is called an $[n, k, d]$ linear code over $\mathbb F_q$. The one such that 
$n = k + d -1$ is a \emph{maximum distance separable} (MDS) code.  Moreover, given a linear code $\CC$ over $\mathbb F_q$, its (Euclidean) dual code is denoted by $\CC^ {\perp}$. 

Linear complementary pairs (LCP) of codes are extensively explored because of their unique algebraic structure and wide application in cryptography. This concept was first introduced by Bhasin et al. in~\cite{NBD15}. LCP of codes over a finite field is further studied in~\cite{BCC14} and~\cite{CG18}.  A pair of linear codes $(\CC,\CD)$ over $\FF_q$ of length $n$ is called LCP if $\CC \oplus \CD = \FF_q^n$, where $\oplus$ represents the direct sum of two subspaces.  The concept of LCP of codes is related to the notion of Linear complementary dual (LCD) codes. Indeed, when $\CD = \CC^\perp$, $\CC$ is an LCD code. The notion of an LCD code was first introduced by James L. Massey in 1992~\cite{Mas92}, long before their recent cryptographic applications. These codes provided an optimum linear coding solution for the two-user binary adder channel. Massey gave a characterization and some constructions of codes with complementary duals. He also showed that LCD codes are asymptotically good. 
In 1994, X. Yang and  L. Massey characterized cyclic LCD codes (see~\cite{YM94}).  Further, in 2018, Mesnager, Tang, and Qi.~\cite{MTQ18} studied algebraic geometry LCD codes over finite fields. In the same paper, they obtained good examples from projective lines, elliptic, and Hermitian curves. The same year, Carlet et al. (~\cite{CMT18}) proved that every nonbinary linear code is equivalent to an LCD code with the same parameters. More precisely, they showed that when $q > 3,$ any linear code over $\FF_q$ is equivalent to an Euclidean LCD code, so when $q > 3,$ $q$-ary Euclidean LCD codes are as good as $q$-ary linear codes. On the other hand, in~\cite{CG18}, Carlet et al. proposed a structure of LCP codes over finite fields and constructed a good example of LCP codes, which is more useful than LCD codes over finite fields. After that, Carlet et al. \cite{CMTQ19} first introduced the notion of $\sigma$-LCD codes over finite fields, which is a generalization of Euclidean and Hermitian LCD codes (See \cite{CMTQ18}). In the same paper, they showed that $\sigma$-LCD codes allow the construction of LCP codes. However, this topic has been popular for its valuable application in the context of masking schemes and robustness against side-channel and fault injection attacks, which are shown in ~\cite{BCC14,CG16,NBD15}. Carlet et al.~\cite{CG18} showed that if the pair $(\CC,\CD)$ is LCP, where $\CC$ and $\CD$ are cyclic codes over a finite field, then $\CC$ and $\CD^\perp$ are equivalent. They further showed that if the length of the codes is relatively prime to the characteristic of the finite field and $\CC$ and $\CD$ both are $2\CD$ cyclic codes, then $\CC$ and $\CD^\perp$ are equivalent. Later, Guneri et al.~\cite{GOS18} extended the same results for linear codes $\CC$ and $\CD$, which are $m\CD$ cyclic codes for some $m \in \mathbb{N}$. In this context, the security parameter for LCP of codes $(\CC,\CD)$ is defined to be the minimum of the minimum distances of $\CC$ and $\CD^\perp$, i.e., it is $\min\{d(\CC), d(\CD^\perp)\}$. This parameter is $d(\CC)$ for the LCD case as $\CD^\perp = \CC$. The aim is to construct LCP codes with significant security parameters to strengthen the system's security.

Inspired by~\cite{MTQ18}, we study the LCP of algebraic geometry codes in this article.  However, it is challenging to obtain explicit constructions as requested by many applications, and only a few known LCP codes have been explicitly constructed. The main objective of this paper is to explicitly construct some classes of  LCP codes from algebraic curves.

This paper is organized as follows. Section~\ref{sec:pre} recalls the basic material of linear codes over a finite field and introduces the necessary background about algebraic geometry codes. Then, we define a complementary pair of algebraic geometry codes over an arbitrary finite field. In Section~\ref{sec:char}, we elaborate on the property of a complementary pair of algebraic geometry codes and obtain some characterization for a complementary pair of codes over a finite field. Further, we present a pair of algebraic geometry codes $(\CC,\CD)$ such that $\CC^\perp$ is equivalent to $\CD$ under specific conditions. We derive a complementary pair of algebraic geometry codes from elliptic curves in Section~\ref{sec:elpt}.
Furthermore, we obtain a complementary pair of algebraic geometry codes from arbitrary algebraic geometry codes in Section~\ref{sec:agc}. Finally, we present a pair of algebraic geometry codes $(\CC,\CD)$ such that $\CC^\perp$ is equivalent to $\CD$.  Given two algebraic geometry codes (with some assumptions on the involved divisors), we shall provide simple ways to give rise to a pair of LCP codes from the initial codes. 

\section{Some preliminaries}\label{sec:pre} 

In this section, we briefly introduce linear codes and algebraic geometry codes over finite fields. Codes with excellent properties have been obtained using techniques and resources from algebraic curves, the so-called algebraic geometry codes.

Throughout the paper, let $\FF_q$ be the finite field with cardinality $q = p^m$ for a prime $p$, and $\FF_q^*$ be the multiplicative group of order $q-1.$ 

\subsection{Background on linear codes}
An $[n, k, d]$ linear code $\mathcal C$ over $\mathbb F_q$  of length $n$  is a linear subspace of  $\FF^n_q$ with dimension $k$ and minimum (Hamming) distance $d$. The minimum distance of an $[n, k, d]$ linear code is bounded by the Singleton bound
\[d\le n+1-k.\] A  code meeting the above bound is called \emph{Maximum Distance Separable} (MDS).
  
 For any two vectors $\bx = (x_1, x_2, \ldots, x_n)$ and $\by = (y_1, y_2, \ldots, y_n)$ in $\FF_q^n$, the inner-product between $\bx$ and $\by$ are defined by $\langle \bx,\by \rangle=\sum\limits_{i=1}^nx_iy_i$.   For a linear code $\CC$ over $\FF_q$ of length $n$, the code $\CC^\perp = \{\bx \in \FF^n_q~|~\langle \bx,\mathbf{c} \rangle = 0 \text{ for all } \mathbf{c} \in \CC\}$ is said to be the dual of $\CC$. For a vector $\ba = (a_1, a_2, \ldots, a_n) \in \FF_q^n$ and a linear code $\CC$, $\ba \CC$ is defined as  $\ba \CC=\{(a_1c_1, a_2c_2, \ldots, a_nc_n)~|~(c_1, c_2, \ldots, c_n)\in \CC\}$.  Note that $\ba\CC$ is linear if $\CC$ is linear. The Euclidean hull of a linear code $\mathcal C$ is defined to be $\text{Hull}_{E}(\mathcal C):= \mathcal C \cap \mathcal C^ {\perp}$
 and $h_{E}(\mathcal C)$ be the dimension of $\text{Hull}_{E}(\mathcal C)$. A linear code $\mathcal C$ over $\mathbb F_{q}$ is called an \emph{LCD code} (or for short, LCD code) if $h_{E}(\mathcal C)=0$.  A pair of linear codes $(\CC,\CD)$ over $\FF_q$ of length $n$ is called LCP if $\CC \oplus \CD = \FF_q^n$, where $\oplus$ represents the direct sum of two subspaces. When $\CD = \CC^\perp$ then, $\CC$ is an LCD code.
 
 For any $\mathbf a=(a_1,\ldots, a_n) \in \mathbb F_q^{n}$ and permutation $\sigma$ of $\{1,2,\ldots, n\}$, we define $\mathcal C_{\mathbf a}$ and $\sigma(\mathcal C)$ as the following linear codes
\begin{align*}
\mathcal C_{\mathbf a}=\{(a_1c_1,\ldots, a_n c_n): (c_1, \ldots, c_n) \in \mathcal C\},\\
\sigma(\mathcal C)=\{(c_{\sigma(1)},\ldots, c_{\sigma(n)}): (c_1, \ldots, c_n) \in \mathcal C\}.
\end{align*}

Two codes $\mathcal C$ and $\mathcal C'$ in $\mathbb F_q^n$ are called \emph{equivalent} if $\mathcal C'=\sigma(\mathcal C_{\mathbf a})$ for some permutation $\sigma$ of $\{1,2,\ldots, n\}$ and $\mathbf a \in \mathbb (\mathbb F_q^*)^n$.  Any $[n,k]$ linear code over a finite field is equivalent to a code generated by a matrix of the form $[I_k: P]$ where $I_k$ denotes the $k\times k$ identity matrix. Notably,  if $\ba = (a_1, a_2, \ldots, a_n) \in \left(\FF_q^*\right)^n$, then $\CC$ is equivalent to $\ba\CC$ and $(\ba\CC)^\perp = \ba^{-1}\CC^\perp$, where $\ba^{-1}=(a_1^{-1},a_2^{-1},\ldots,a_n^{-1})$.

\subsection{Background on algebraic geometry codes}

The following presents some basic definitions and results of algebraic geometry codes needed to derive our main results. We refer to~\cite{Sti08} for more detailed information on algebraic geometry codes.  These codes are important in the context of LCD codes. Notably, Jin and Xing \cite{JX-LCD} showed that an algebraic geometry code over $\mathbb{F}_2^m$ ($m \ge  7$) is equivalent to an LCD code. Consequently, they proved that algebraic geometry codes with complementary duals exceed the well-known Gilbert-Varshamov asymptotic bound.

Let $\CX$ be a smooth projective curve of genus $g$ over $\FF_q$. The function field corresponding to $\CX$ is denoted by $\FF_q(\CX)$. The maximal ideal in a valuation ring $O$ of $\FF_q(\CX)$ is called place $P$ of $\FF_q(\CX)$. Note that $O/P$ is isomorphic to an extended field over $\FF_q$. The degree of a place $P$ is defined by $\deg(P) = [O/P:\FF_q]$. Note that a place is {\em rational} if the degree of the place is one.
We denote $\mathbb{P}$ as the set of all places of $\FF_q(\CX)$. A {\em divisor} $D$ is defined by $$D = \sum\limits_{P \in \mathbb{P}} m_P P, \text{ all } m_P = 0 \text{ except finitely many}.$$
The {\em support} of a divisor $D$ is defined as $\spt(D) = \{P \in \mathbb{P}~|~m_P \neq 0\}$ and the {\em degree} of $D$ is defined as $\deg(D) = \sum\limits_{P \in \spt(D)}m_P\deg(P)$. If the places are rational (i.e., their degrees are one), $\deg(D) = \sum\limits_{P \in \spt(D)}m_P$.
Furthermore, $D = \sum\limits_{P\in\mathbb{P}}m_PP \leq  D' = \sum\limits_{P\in\mathbb{P}}n_pP$ if and only if $m_P \leq n_P$, for all $P\in\mathbb{P}$. Two divisors $D$ and $D'$ are {\em equivalent} if there exists $f$ in $\FF_q(\CX)$ such that $(f) = D - D'$ and is denoted as $D \sim D'$.

For each place $P$, the corresponding discrete valuation is denoted by $v_P$. 
The {\em principal divisor} of a non-zero function $f$ is defined by $(f) = \sum\limits_{P\in \mathbb{P}}v_P(f)P$. 

The Riemann-Roch space associated to a divisor $G$ is 
$$\CL(G)=\{f\in\FF_q(\CX)\setminus\{0\}~:~(f)\geq -G \}\cup\{0\}.$$ 
Note that $\CL(G)$ is a finite-dimensional vector space over $\FF_q$, and its dimension is denoted by $\ell(G)$.\\
Let $P_1, P_2, \ldots, P_n$ be rational points on $\CX$ and $G, H$ and $D$ be three divisors such that $D = P_1 + P_2 + \cdots + P_n$ and $\spt(G) \cap \spt(D) = \spt(H)\cap \spt(D) = \emptyset$. Then, the image of the mapping 
$$\varphi:~\CL(G)\rightarrow \FF_q^n,~f\mapsto(f(P_1),f(P_2),\dots,f(P_n)),$$
is a linear code over $\FF_q$. This code is called an algebraic geometry code, and it is defined by 
$$C_{\CL}(D,G)=\{(f(P_1),f(P_2),\dots,f(P_n))~|~f\in\CL(G)\}.$$
If $2g-2 < \deg(G) < n$, then $\dim(C_{\CL}(D,G)) = \ell(G) = \deg(G) - g + 1$.

Let $\Omega$ be the module of Weil differentials of $\FF_q(\CX)$ and is defined by $$\Omega:=\{fdx~:~f\in \FF_q(\CX)\}.$$
$\Omega$ is a one-dimensional vector space over $\FF_q$.
For a place $P$ and a function $t$ with the condition $v_P(t)=1$, we define $v_P(fdt) = v_P(f)$. The canonical divisor is defined as 
$$(\omega) = \sum\limits_{P \in \mathbb{P}} v_P(\omega)P \text{ all } v_P = 0 \text{ except finitely many},$$
where $\omega \in \Omega$. We define a set by 
$$\Omega(G):=\{\omega\in\Omega : (\omega) \geq G\} \cup\{0\}.$$
Note that $\Omega(G)$ is a finite-dimensional vector space over $\FF_q$, and its dimension is $i(G)$. $G$ is called a {\em non-special divisor} if $i(G) = 0$. Furthermore, if $G$ and $H$ are two divisors such that $G$ is a non-special and $G \leq H$, then $H$ is also non-special.

Below is a celebrated famous result~\cite{Sti08}, known as the Riemann-Roch theorem, a central influence in algebraic geometry with applications in other areas and the key to several developments in coding theory.

\begin{thm}\label{Riemann-Roch}
Let $G$ be a divisor on a smooth projective curve $\CX$ of genus $g$ over $\mathbb F_q$. Then, for any canonical divisor $K$
\begin{align*}
\ell(G) -i(G)=deg(G)+1-g \text{~~~~and~~~~} i(G)=l(K-G).
\end{align*}
\end{thm}

If $\deg(G) > 2g-2$, i.e., $G$ is non-special, then $i(G) = 0$~\cite{Sti08}. In that case,
\begin{equation}\label{eqn:ig}
\ell(G) = \deg(G) - g + 1.
\end{equation}
Finally, we define another set $C_{\Omega}(D,G)$ as
$$C_{\Omega}(D,G) = \{\textit{res}_{P_1}(\omega), \textit{res}_{P_2}(\omega), \ldots, \textit{res}_{P_n}(\omega)~|~\omega\in\Omega(G-D)\}.$$
Note that the dual of $C_{\CL}(D,G)$ is $C_{\Omega}(D,G)$~\cite{Sti08}.

An algebraic geometry code $C_{\CL}(D,G)$ associating with divisors $G$ and $D$ over the projective line is said to be \emph{rational}. In particular, BCH and Goppa codes can be described using rational AG codes. All the well-known generalized Reed-Solomon codes and extended generalized Reed-Solomon codes can be defined under the framework of algebraic geometry codes.

\section{Characterization of LCP of codes over from algebraic curves and its consequences}\label{sec:char}

Let $\CX$ be a smooth projective curve of genus $g$ over $\FF_q$. Throughout the paper, we consider $P_1, P_2, \ldots, P_n$ are rational points on $\CX$ and $G, H$ and $D$ are divisors such that $D=P_1 + P_2 + \cdots + P_n$ and $\spt(G)\cap \spt(D) = \spt(H) \cap \spt(D) = \emptyset$. According to the definition,  a pair $(C_{\CL}(D,G), C_{\CL}(D,H))$ of algebraic geometry codes is a linear complementary pair (LCP) of codes over $\FF_q$ if $C_{\CL}(D,G)\oplus C_{\CL}(D,H)=\FF_q^n$. In other words, a pair $(C_{\CL}(D,G), C_{\CL}(D,H))$ of algebraic geometry codes is LCP if and only if $C_{\CL}(D,G) \cap C_{\CL}(D,H) = \{0\}$ and $\dim_{\FF_q}(C_{\CL}(D,G)) + \dim_{\FF_q}(C_{\CL}(D,H)) = n$. In addition, if $2g-2 < \deg(G) < n$ and $2g-2 < \deg(H) <n$, then we have
\begin{equation}\label{eq-3.1}
\ell(G)+\ell(H)=n.
\end{equation}
For the divisors, $G = \sum\limits_{i=1}^n m_iQ_i$ and $H = \sum\limits_{i=1}^nm_i'Q_i$, we define 
\begin{equation}\label{eq-3.2}
\gcd(G,H) = \sum_{i=1}^n \min(m_i,m_i')Q_i,~~~\lmd(G,H)=\sum_{i=1}^n \max(m_i,m_i')Q_i.
\end{equation}     
As the algebraic geometry code is defined, we have 
$$C_{\CL}(D,\gcd(G,H)) = \{(f(P_1),f(P_2),\dots,f(P_n))|~f\in\CL(\gcd(G,H))\} \text{~~~~and }$$ 
$$C_{\CL}(D,\lmd(G,H)) = \{(f(P_1),f(P_2),\dots,f(P_n))|~f\in\CL(\lmd(G,H))\}.~~~~~~~~~$$
\begin{lemma}
Let $G = \sum\limits_{i=1}^n m_iQ_i, H = \sum\limits_{i=1}^nm_i'Q_i$ and $D = P_1 + P_2 + \cdots + P_n$ be three divisors over a smooth projective curve $\CX$ of genus $g$ over $\FF_q$. If $\spt(G) \cap \spt(D) = \spt(H) \cap \spt(D) = \emptyset$. Then $\ell(\lmd(G,H)-D) \leq \ell(\lmd(G,H))$ and $|\ell(\lmd(G,H)-D)-\ell(\lmd(G,H))|\leq n$.
\end{lemma}
\begin{proof}
To the place $P,$ we associate discrete valuation $v_p$. Now,
\[v_p = \left\{\begin{array}{ll}
-1 & \textrm{if }P=P_i,~1\leq i\leq n;\\
\textit{max}(m_i,m_i') & \textrm{if }P=Q_i,~ \text{from Eq. \eqref{eq-3.2}}. 
\end{array}\right.\]
It follows that $\ell(\lmd(G,H)-D)\leq\ell(\lmd)(G,H).$ 
\begin{eqnarray*}
\text{Further, }&& \dim\dfrac{\CL(\lmd(G,H))}{\CL(\lmd(G,H)-D)} \leq \deg((\lmd(G,H)))-\deg((\lmd(G,H))-D) \\
& \implies & \dim(\CL(\lmd(G,H))) - \dim(\CL(\lmd(G,H)-D)) \leq n \\
& \implies & 0 \leq \ell(\lmd((G,H))-\ell(\lmd(G,H) - D) \leq n \\
& \implies & |\ell(\lmd((G,H)-D)-\ell(\lmd(G,H))| \leq n.
\end{eqnarray*}
\end{proof}
\begin{lemma}\label{lm-3.2}
If $C_{\CL}(D,G)$ and $C_{\CL}(D,H)$ are two algebraic geometry codes, then
\begin{enumerate}
\item[(1).] $C_{\CL}(D,\gcd(G,H))\subseteq C_{\CL}(D,G)\cap C_{\CL}(D,H)$;
\item[(2).] $C_{\CL}(D,G)+ C_{\CL}(D,H)\subseteq C_{\CL}(D,\lmd(G,H))$. 
\end{enumerate}
\end{lemma}
\begin{proof}
For the proof of $(1)$, let $\bx \in C_{\CL}(D,\gcd(G,H))$, then there exists $f\in\CL(\gcd(G,H))$ such that
$$\bx = (f(P_1), f(P_2), \ldots, f(P_n)).$$
Since $f\in\CL(\gcd(G,H))$, $(f)\geq -\gcd(G,H)$ and this implies that $(f) \geq -G$ and $(f) \geq -H$.  Therefore, $\bx = (f(P_1), f(P_2), \ldots, f(P_n))\in C_{\CL}(D,G) \cap C_{\CL}(D,H)$. Hence, $C_{\CL}(D,\gcd(G,H)) \subseteq C_{\CL}(D,G)\cap C_{\CL}(D,H)$.

For the proof of $(2)$, let $\by \in C_{\CL}(D,G) + C_{\CL}(D,H)$, then there exist $f \in \CL(G)$ and $g \in \CL(H)$ such that
$$(f(P_1), f(P_2), \ldots, f(P_n))\in C_{\CL}(D,G) \text{ and } (g(P_1), g(P_2), \ldots, g(P_n)) \in C_{\CL}(D,H),$$ 
with $\by = ((f+g)(P_1), (f+g)(P_2), \ldots, (f+g)(P_n))$. To each place $P$, we associate discrete valuation $v_p$. Now, $v_P(f+g) \geq \min\{v_p(f), v_P(g)\} \geq -v_p(\lmd(G,H))$, which implies that $f + g \in \CL(\lmd(G,H))$. Hence, $\by = ((f+g)(P_1), (f+g)(P_2), \ldots, (f+g)(P_n)) \in C_{\CL}(D,\lmd(G,H))$.
Therefore, $C_{\CL}(D,G)+ C_{\CL}(D,H) \subseteq C_{\CL}(D,\lmd(G,H))$.
\end{proof}
Next, we present a theorem which will be used to construct an LCP of algebraic geometry codes from a smooth projective curve $\CX$ of genus $g = 0$ over $\FF_q$. 

\begin{thm}\label{th-3.1}
Let $C_{\CL}(D,G)$ and $C_{\CL}(D,H)$ be two algebraic geometry codes over $\FF_q$ of length $n$ with genus $g = 0$, such that $\ell(G) + \ell(H) = n$ and $2g-2 < \deg(G),\deg(H) < n$. Then the pair $(C_{\CL}(D,G), C_{\CL}(D,H))$ is LCP if the degree of the divisor $\gcd(G,H)$ is $g-1$.
\end{thm}

\begin{proof}
By assumption, $\deg(\gcd(G,H)) = g-1 > 2g-2$ as $g = 0$. Therefore, $\gcd(G,H)$ is a non-special divisor, i.e., $i(\gcd(G,H)) = 0$. Then $\ell(\gcd(G,H)) = \deg(\gcd(G,H)) - g + 1 = 0$, i.e.,
\begin{equation}\label{eq-11}
\ell(\gcd(G,H))=0.
\end{equation}
As $\deg(\lmd(G,H)) \geq \deg(\gcd(G,H)) > 2g-2$, $\lmd(G,H)$ is also a non-special divisor.
By assumption, $\ell(G) + \ell(H) = n$, which implies that $\deg(G) - g + 1 + \deg(H) - g + 1 = n \implies \deg(G) + \deg(H) = n + 2g - 2$ as $G$ and $H$ are non-special. It is noted that $\deg(\lmd(G,H)) + \deg(\gcd(G,H)) = \deg(G) + \deg(H) = n + 2g - 2$. Therefore, $\deg(\lmd(G,H)) = n + g - 1$ as $\deg(\gcd(G,H)) = g - 1$. Hence, $\deg(\lmd(G,H)-D) = g-1 > 2g-2$. Hence, $\lmd(G,H)-D$ is non-special. Then $\ell(\lmd(G,H)-D) = \deg(\lmd(G,H)-D) - g + 1$ and that implies
\begin{equation}\label{eq-12}
\ell(\lmd(G,H)-D) = 0.
\end{equation}
It is enough to prove that  $C_{\CL}(D,G) \cap C_{\CL}(D,H) = \{0\}$. Let $\bx \in C_{\CL}(D,G)\cap C_{\CL}(D,H)$. Then there exist $f \in \CL(G)$ and $g \in \CL(H)$ such that $\bx = (f(P_1), f(P_2), \ldots, f(P_n)) = (g(P_1), g(P_2), \ldots, g(P_n))$. Then $(f-g)(P_i) = 0$, for all $1\leq i\leq n$.
If $f = g$, then $f\in \CL(G)\cap \CL(H)$. That implies $f \in \CL(\gcd(G,H))$.
By Equation~\eqref{eq-11}, we get $f = 0$. If $f \neq g$, it can be derived that $f-g \in \CL(\lmd(G,H)-D)$ as $(f-g)(P_i) = 0$, for all $1 \leq i \leq n$. Then by Equation~\eqref{eq-12}, we get $f-g = 0$ i.e., $f = g$, which is a contradiction.
Hence, $C_{\CL}(D,G) \cap C_{\CL}(D,H) = \{0\}$.
\end{proof}
The following result shows that, under some assumptions, when two algebraic geometry codes form a pair of LCP codes, the degree of the Greatest Commun Divisor of the involved divisors can be easily computed in terms of the genus.

\begin{prop}
Let $C_{\CL}(D,G)$ and $ C_{\CL}(D,H)$ be two algebraic geometry codes over $\FF_q$ of length $n$ with genus $g = 0$ and $2g-2 < \deg(G),\deg(H) < n$. If the pair $(C_{\CL}(D,G),C_{\CL}(D,H))$ is LCP and $\gcd(G,H)$ is non-special, then $\deg(\gcd(G,H)) = g - 1$.  
\end{prop}
\begin{proof}
By assumption, the pair $(C_{\CL}(D,G), C_{\CL}(D,H))$ is LCP, i.e.,   
$C_{\CL}(D,G) \cap C_{\CL}(D,H) = \{0\}.$  By the Lemma~\ref{lm-3.2}, we have $C_{\CL}(D,\gcd(G,H))\subseteq C_{\CL}(D,G)\cap C_{\CL}(D,H)=\{0\}.$ Therefore, $C_{\CL}(D,\gcd(G,H))=\{0\},$ which gives that $$\ell(\gcd(G,H))-\ell(\gcd(G,H)-D)=0.$$
It is easy to see that $\deg(\gcd(G,H)-D)<0,$ as $\deg(\gcd(G,H))\leq\deg(G)<\deg(D)=n.$ Therefore, $$\ell(\gcd(G,H))=\ell(\gcd(G,H)-D)=0.$$
As $\gcd(G,H)$ is non-special, from Equation~\eqref{eqn:ig}, $\deg(\gcd(G,H)) =  g - 1$.
\end{proof}
\begin{example}
Let $\CX$ be the projective line over $\FF_q$ and $\CO, \CQ$ are rational points on $\CX$ such that $\CO$ is the point at infinity, $\CQ$ is the original point. Let us consider $G = (q-s-2)\CO - (s+1)\CQ, H = s\CO + s\CQ$ and $D=\CX(\FF_q)\setminus\{\CO,\CQ\}$ be the divisors, where $0 < s \leq \frac{q-2}{2}$. Then $\gcd(G,H) = s\CO-(s+1)\CQ$ and $\deg(\gcd(G,H)) = -1$. Hence, by Theorem~\ref{th-3.1}, we have that $(C_{\CL}(D,G), C_{\CL}(D,H))$ forms an LCP of codes.
\end{example}
In the following, we derive from constructing LCP of algebraic geometry codes from a smooth projective curve $\CX$ of genus $g\neq 0$ over $\FF_q$.
\begin{thm}\label{th-3.2}
Let $C_{\CL}(D,G)$ and $ C_{\CL}(D,H)$ be two algebraic geometry codes over $\FF_q$ of length $n$ with genus $g\neq 0$, such that $\ell(G) + \ell(H) = n$ and $2g-2 < \deg(G),\deg(H) < n$. Then the pair $(C_{\CL}(D,G), C_{\CL}(D,H))$ is LCP if $\gcd(G,H)$ is a non-special divisor of degree $g-1$ and $\lmd(G,H)-D$ is a non-special divisor.
\end{thm}
\begin{proof}
It is enough to prove that  $C_{\CL}(D,G) \cap C_{\CL}(D,H) = \{0\}$.  Let $\bx \in  C_{\CL}(D,G) \cap C_{\CL}(D,H)$. Then there exist some $f \in \CL(G)$ and $g \in \CL(H)$ such that $\bx = (f(P_1), f(P_2), \ldots, f(P_n)) = (g(P_1), g(P_2), \ldots, g(P_n))$. This implies $(f-g)(P_i) = 0$, for all $1 \leq i \leq n$. If $f = g$ then $f \in \CL(G)\cap \CL(H)$. That implies
\begin{equation}\label{eq-3.4}
f\in\CL(\gcd(G,H)).
\end{equation}
If $f \neq g$, then (as proved in the Theorem~\ref{th-3.1})
\begin{equation}\label{eq-3.5}
f-g\in \CL(\lmd(G,H)-D).
\end{equation}
As $\gcd(G,H)$ is a non-special divisor of degree $g-1$, $\ell(\gcd(G,H)) = 0$ (from Equation~\eqref{eqn:ig}). Then, by Equation~\eqref{eq-3.4}, we have $f = 0$. Further, $\ell(G)+\ell(H) = n$ results that $\deg(G) + \deg(H) = n + 2g - 2$. Then $\deg(\lmd(G,H)) + \deg(\lmd(G,H)) = \deg(G)+ \deg(H) = n + 2g - 2$ and that implies $\deg(\lmd(G,H)) = n + g - 1$. Hence
$$\deg(\lmd(G,H)-D) = g-1.$$
Therefore, $\ell(\lmd(G,H)-D) = 0$ as $\lmd(G,H)-D$ is a non-special divisor of degree $g-1$. Then, by Equation~\eqref{eq-3.5}, we have $f = g,$ which contradicts that $f\neq g$. Therefore, $C_{\CL}(D,G)\cap C_{\CL}(D,H)=\{0\}$ and that implies the pair $(C_{\CL}(D,G),C_{\CL}(D,H))$ is LCP.
\end{proof}

The existence of a non-special divisor $G$ with degree $g-1$ was presented in~\cite{BL06}.

\begin{thm}
Let $C_{\CL}(D,G)$ and $ C_{\CL}(D,H)$ be two algebraic geometry codes over $\FF_q$ of length $n$ with genus $g \neq 0$ and $2g-2 < \deg(G),\deg(H) < n$. If the pair $(C_{\CL}(D,G), C_{\CL}(D,H))$ is LCP and $\deg(\gcd(G,H)) = g - 1$ then $\gcd(G,H)$ and  $\lmd(G,H)-D$ are non-special divisors.
\end{thm}
\begin{proof}
As the pair $(C_{\CL}(D,G), C_{\CL}(D,H))$ is LCP, $C_{\CL}(D,G) \cap C_{\CL}(D,H) = \{0\}$ and $C_{\CL}(D,G) + C_{\CL}(D,H) = \FF_q^n$. That is, $\dim(C_{\CL}(D,G)) + \dim(C_{\CL}(D,H)) = \ell(G) + \ell(H) = n$. Then $\deg(G) +\deg(H) = n + 2g - 2$. Therefore,
$$\deg(\lmd(G,H)) + \deg(\lmd(G,H)) = \deg(G) + \deg(H) = n + 2g - 2.$$
By Lemma~\ref{lm-3.2}, we have 
\begin{eqnarray*}
C_{\CL}(D,\gcd(G,H)) & \subseteq & C_{\CL}(D,G)\cap C_{\CL}(D,G)=\{0\} \text{ and } \\
C_{\CL}(D,\lmd(G,H)) & \supseteq & C_{\CL}(D,G) + C_{\CL}(D,G) = \FF_q^n.
\end{eqnarray*}
Hence, $C_{\CL}(D,\gcd(G,H))=\{0\}~\text{and}~C_{\CL}(D,\lmd(G,H))=\FF_q^n$.\\
It implies that $\ell(\gcd(G,H)) - \ell(\gcd(G,H)-D) = 0$ and $\ell(\lmd(G,H)) -\ell(\lmd(G,H)-D) = n$. As $\deg(\gcd(G,H)-D) < 0$, $\ell(\gcd(G,H)) = 0$. By assumption, $\deg(\gcd(G,H)) = g-1$, it follows that $\gcd(G,H)$ is non-special and $\deg(\lmd(G,H)) = n+g-1$. Therefore, $\lmd(G,H)$ is also non-special and it follows that $\ell(\lmd(G,H)) = n$. Hence, $\ell(\lmd(G,H)-D) = 0$. As $\deg(\lmd(G,H)-D) = g-1$, it follows that $\lmd(G,H)-D$ is a non-special. 
\end{proof}
\begin{example}
Consider the projective curve $\CX:~Y^2Z+YZ^2=X^3$ of genus $1$ over $\FF_4:=\{0,1,\omega,\omega^2\}$.
Here $\CX(\FF_4) = \{\CO, \CQ, P_1, P_2, \ldots, P_7\} = \{(0,1,0), (0,0,1),(0,1,1), (\omega,\omega,1), (\omega,\omega^2,1), \\ (\omega^2,\omega,1), (\omega^2,\omega^2,1), (1,\omega,1), (1,\omega^2,1)\}$.\\
Further, consider $G = 6\CO - 2\CQ, H = 2\CO + \CQ$ and $D = \sum\limits_{i=1}^7 P_i$. Then $\gcd(G,H)=2\CO - 2\CQ$ and $\lmd(G,H)-D = 6\CO + \CQ - \sum\limits_{i=1}^7 P_i$ are non-special divisor. Then by Theorem~\ref{th-3.2}, the pair $(C_{\CL}(D,G), C_{\CL}(D,H))$ forms an LCP.
\end{example}

Next, we study the relationship among $C_{\CL}(D,H)$ and $C_{\CL}(D,G)^\perp$ where the pair $(C_{\CL}(D,G), C_{\CL}(D,H))$ forms  an LCP of codes.  Given two algebraic geometry codes (with some assumptions on the involved divisors), the two following results provide a necessary condition so that both codes form a pair of LCP codes.
\begin{thm}
Let $C_{\CL}(D,G)$ and $ C_{\CL}(D,H)$ be two algebraic geometry codes over $\FF_q$ of length $n$ with genus $g = 0$ and $2g-2 < \deg(G),\deg(H) < n$. If the pair $(C_{\CL}(D,G), C_{\CL}(D,H))$ is an LCP, then $C_{\CL}(D,H)$ and $C_{\Omega}(D,G) = C_{\CL}(D,G)^\perp$ are equivalent. 
\end{thm}
\begin{proof}
Suppose that $C_{\CL}(D,G)$ and $C_{\CL}(D,H)$ are $[n,k,d]$ and $[n,k',d']$ codes, respectively. Let $C_{\Omega}(D,G) = C_{\CL}(D,G)^\perp$ is an $[n,n-k,d^\perp]$ code. Now, we will prove $k' = n-k$ and $d' = d^\perp$.
As the pair $(C_{\CL}(D,G),C_{\CL}(D,H))$ is LCP, $k'= n-k$.

Here, $d \geq n - \deg(G)$ and $d' \geq n - \deg(H)$~\cite[Theorem II.2.2]{Sti08}.
As $\deg(G) > 2g-2$ and $\deg(H) > 2g-2$, $k = \deg(G) - g + 1$ and $k' = \deg(H) - g + 1$. Then, from the Singleton bound and $g = 0$, we have $d \leq n-\deg(G)$ and $d' \leq n-\deg(H)$. Hence, $d = n - \deg(G)$ and $d' = n - \deg(H)$ i.e., both $C_{\CL}(D,G)$ and $C_{\CL}(D,H)$ are MDS codes. 

As $k' = n - k = n - \deg(G)  +g - 1$, $\deg(H) = k' + g - 1 = n-\deg(G) + 2g - 2 \implies d' = n - \deg(H) = \deg(G) - 2g + 2$.
Since $C_{\Omega}(D,G)$ is the dual of $C_{\CL}(D,G)$, $C_{\Omega}(D,G)$ is an MDS code, and $\dim(C_{\Omega}(D,G)) = n-k = n - \deg(G) + g - 1$. 
Then $d^\perp = n-(n-k)+1=\deg(G) - 2g + 2 = d'$. Hence,  $C_{\CL}(D,H)$ and $C_{\Omega}(D,G)$  are equivalent. 
\end{proof}

\begin{thm}
Let $C_{\CL}(D,G)$ and $ C_{\CL}(D,H)$ be two algebraic geometry codes over $\FF_q$ of length $n$ with genus $g\neq 0$ and $2g-2 < \deg(G), \deg(H) < n$. Let $\omega$ be a Weil differential such that $(\omega)=G + H' - D$ for some divisor $H'$. If the pair $(C_{\CL}(D,G),C_{\CL}(D,H))$ is LCP and $H\sim H'$, then $C_{\CL}(D,H)$ and $C_{\Omega}(D,G) = C_{\CL}(D,G)^\perp$  are equivalent. 
\end{thm}  
\begin{proof}
Since $\omega$ be a Weil differential such that $(\omega)=G + H'- D$ for some divisor $H'$, $$C_{\CL}(D,G)^\perp=C_{\Omega}(D,G)=C_{\CL}(D,H').$$ As $H \sim  H'$, $\deg(H)=\deg(H')$ and $\ell(H)=\ell(H')$. It follows that $C_{\CL}(D,H)$ and $C_{\Omega}(D,G)$  are equivalent.
\end{proof}

\section{LCP of codes from elliptic curves}\label{sec:elpt}
The main purpose of this section is to construct some LCP of algebraic geometry codes from elliptic curves. It is well-known that an elliptic curve is an algebraic curve with genus $1$ and a divisor of an elliptic curve $\CE$ is not a principal if and only if every $0$ degree divisor is a non-special. Then, the following result can be derived from Theorem ~\ref{th-3.2}. 

\begin{thm}\label{th-4.1}
Let $G, H$ and $D = P_1 + P_2 + \cdots + P_n$ be three divisors over an elliptic curve $\CE$ over $\FF_q$ with $\spt(G) \cap \spt(D) = \spt(H) \cap \spt(D) = \emptyset$ and $\dim(C_{\CL}(D,G)) + \dim(C_{\CL}(D,H)) = n$, where $0 < \deg(G), \deg(H) < n$. Then, the pair $(C_{\CL}(D,G),C_{\CL}(D,H))$ is LCP if $\gcd(G,H)$ is not a principal divisor of degree $0$ and  $\lmd(G,H)-D$ is not a principal divisor.
\end{thm}

Now, we shall use elliptic curves in Weierstrass form to construct the LCP of codes. For $a, b, c \in \FF_q$, we denote an affine elliptic curve by the equation
$$\CE_{a,b,c}:~y^2 + ay = x^3 + bx + c,$$
and the total number of rational points on $\CE_{a,b.c}$ by $\mathcal{N}$. Let $\mathcal{S}$ be the set of $x$-components of the affine points on $\CE_{a,b,c}$, i.e., $$\mathcal{S}:=\{\alpha \in \FF_q~|~\exists \beta\in\FF_q~\text{such~that}~\beta^2 + a\beta = \alpha^3 + b\alpha + c\}.$$
For $a = 1$, any $\alpha \in \mathcal{S}$ gives exactly two points with $x$-component $\alpha$, and we denote these two points corresponding to $\alpha$ by $P_{\alpha}^{+}$ and $P_{\alpha}^{-}$. Let denote the point at infinity as $\CO$. Then the set of all rational points of $\CE_{1,b,c}$ over $\FF_q$ are $$\{P_{\alpha}^+~|~\alpha \in \mathcal{S}\} \cup \{P_{\alpha}^-~|~\alpha \in \mathcal{S}\} \cup \{\CO\}.$$
For any positive integer $r$, we denote the set $\CE[r]$ as $$\CE[r]:=\{P~|~\underbrace{P \oplus P \oplus \cdots \oplus P}_r = \CO\}.$$ We refer to \cite{Sil85} for more details on elliptic curves.
\begin{thm}
Let $\{\alpha_0, \alpha_1, \ldots, \alpha_s\} \subseteq \mathcal{S}$ on $\CE_{1,b,c}$ and $D = \sum_{i=1}^s P^{+}_{\alpha_i} + \sum_{i=1}^s P^{-}_{\alpha_i}, G = r\CO + rP^{-}_{\alpha_0}$ and $H = (2s-r)\CO - rP^{-}_{\alpha_0}$ with $P^{-}_{\alpha_0}\notin E[r]$ and $0 < r < s$. Then, the pair $(C_{L}(D,G), C_{L}(D,H))$ is an LCP of codes.
\end{thm}
\begin{proof}
Here, $\gcd(G,H) = r\CO-rP^{-}_{\alpha_0}$. Then $\deg(\gcd(G,H))=0$. Further, $\gcd(G,H)$ is not a principal as $rP^{-}_{\alpha_0}\neq \CO$. Similarly, $\deg(\lmd(G,H)-D)=0$, and $\deg(\lmd(G,H)-D)$ is not a principal divisor. Then, by Theorem~\ref{th-4.1}, we obtain the required result.
\end{proof}
\begin{example}
Consider the elliptic curve $$\CE~:~y^2+y=x^3+x+1.$$  Then the set of all rational points over $\FF_8$ are $\{\CO=(0,1,0),~P_{1}^{+}=(\omega,0,1),~P_{1}^{-}=(\omega,1,1),~ P_{2}^{+}=(\omega^2,0,1),~P_{2}^{-}=(\omega^2,1,1),~P_{3}^{+}=(\omega^3,\omega,1),~P_{3}^{-}=(\omega^3,\omega^3,1),~P_{4}^{+}=(\omega^4,0,1),~P_{4}^{-}=(\omega^4,0,1),~P_{5}^{+}=(\omega^5,\omega^4,1),~P_{5}^{-}=(\omega^5,\omega^5,1),~P_{6}^{+}=(\omega^6,\omega^2,1),~P_{6}^{-}=(\omega^6,\omega^6,1)\},$ where $\omega$ is a primitive element of $\FF_8$. Here $s = 6$. Consider $D = \sum\limits_{i=2}^6(P^+_{i}+P^{-}_{i}), G = 4\CO+4P^{-}_{1}$ and $H=8\CO-4P^{-}_{1}$. It is easy to see that $4P_{1}^{-}\neq \CO.$ Now, $\gcd(G,H)=4\CO-4P^{-}_{1}$ is not a principal divisor. Similarly, $\lmd(G,H)=8\CO+4P^{-}_{1}-\sum_{i=2}^{6}(P^{+}_{i}+P^{-}_{i})$ is not principal. Hence, the pair $(C_{\CL}(D,G),C_{\CL}(D,H))$ forms LCP.
\end{example}

\section{LCP codes from arbitrary algebraic geometry codes}\label{sec:agc}
This section aims to construct LCP codes from a given algebraic geometry code. At first, we deal with algebraic geometry codes of even length. 
\begin{thm}\label{th-a}
Let $C_{\CL}(D,G)$ be an $[n,\frac{n}{2},d]$ algebraic geometry code of even length, where $G$ and $D = P_1 + P_2 + \cdots + P_n$ are divisors of a smooth projective curve $\CX$ of genus $g$ with  $2g-2 < \deg(G) < n$ and $\spt(G)\cap \spt(D) = \emptyset$. If $h \in \FF_q(\CX)$ is a function such that $(h(P_1), h(P_2), \ldots, h(P_n)) \in \FF_q^n \setminus \{(0, 0, \ldots, 0), (1, 1, \ldots, 1)\}$ and $hf \notin \CL(G)$ for all $f \in \CL(G)$. Then there exists $\ba \in \FF_q^n$ such that $(C_{\CL}(D,G), \ba C_{\CL}(D,G))$ forms an LCP of codes.
\end{thm}
\begin{proof}
Let us assume that $G' = G-(h).$  Now, we define a mapping $\phi:\CL(G)\rightarrow\CL(G')$ by $f\mapsto fh$ for all $f\in\CL(G).$ It can be checked that $\phi$ is bijective as $\FF_q(\CX)$ is a field. Hence, $\ell(G) = \ell(G')$ and 
\begin{align*}
C_{\CL}(D,G')  & = \{((hf)(P_1), (hf)(P_2), \ldots, (hf)(P_n))~|~f\in\CL(G)\}\\
               & = \ba\{(f(P_1), f(P_2), \ldots, f(P_n))~|~f\in\CL(G)\}\\
               &=\ba C_{\CL}(D,G),
\end{align*}
where $\ba =(h(P_1), h(P_2), \ldots, h(P_n)$.
Now $\dim(\ba C_{\CL}(D,G)) = \dim( C_{\CL}(D,G')) = \ell(G') = \ell(G)$.
Hence, $\dim(C_{\CL}(D,G))+\dim(\ba C_{\CL}(D,G))=n$.
Now, we will prove that $C_{\CL}(D,G) \cap C_{\CL}(D,G') = \{0\}$. If $\bx \in C_{\CL}(D,G) \cap C_{\CL}(D,G')$ and nonzero, then there exist $f \in \CL(G)$ and $g \in \CL(G')$ such that $\bx = (f(P_1), f(P_2),\ldots, f(P_n)) = (g(P_1), g(P_2),\ldots, g(P_n))$. Since $\phi$ is bijective and $g \in \CL(G')$, then there exists $f' \in \CL(G)$ such that $g = hf'$. Then 
$$(f(P_1), f(P_2), \ldots, f(P_n) = ((hf')(P_1), (hf')(P_2), \ldots, (hf')(P_n))$$
$$\implies ((f-hf')(P_1), (f-hf')(P_2), \ldots, (f-hf')(P_n)) = (0, 0, \ldots, 0).$$
That implies, $f-hf'\in \CL(\lmd(G,G')-D) = \CL(G-D)$ as $G=G'+(h)$. Since $2g-2 < \deg(G) < n$, $\CL(G-D) = 0$. Therefore, $hf' = f \in \CL(G)$, a contradiction. Hence, $C_{\CL}(D,G)\cap C_{\CL}(D,G')=\{0\}$. 
\end{proof}
\begin{example}
Let $\CX:~Y^2Z+YZ^2=X^3$ be the projective curve of genus $1$ over $\FF_4:=\{0,1,\omega,\omega^2\}$. Here, $\CX(\FF_4) = \{\CO, \CQ, P_1, P_2, \ldots, P_7\} = \{(0,1,0), (0,0,1), (0,1,1), (\omega,\omega,1), (\omega,\omega^2,1), (\omega^2,\omega,1), \\ (\omega^2,\omega^2,1), (1,\omega,1), (1,\omega^2,1)\}$.\\
Let us consider $G = 2\CO + \CQ$ and $D = \sum\limits_{i=2}^7P_i$. Note that, $\{1, \left(\frac{X}{Z}\right) = \CQ - 2\CO + P_1, \left(\frac{Y+Z}{X}\right) = -\CQ - \CO + 2P_1\}$ is a basis of $\CL(G)$.
The generator matrix of $C_\CL(D,G)$ can be obtained by evaluating the functions in $\{1, \left(\frac{X}{Z}\right), \left(\frac{Y+Z}{X}\right)\}$ at $\{P_2, P_3,\ldots, P_7\}$, i.e.,
$$\left(\begin{array}{cccccc}
1 & 1 & 1 & 1 & 1 & 1  \\
\omega & \omega & \omega^2 & \omega^2 & 1 & 1  \\
\omega & 1 & 1 & \omega^2 & \omega^2 & \omega  \\
\end{array}\right).$$
It can be checked that $C_\CL(D,G)$ is a self-dual algebraic geometry code. 
Let us choose $\left( h\right) = \left( \frac{X+Y+Z}{Z}\right) = P_1 + P_6 + P_7 - 3\CO$.
It is clear that $h \notin \CL(G)$ and also $hf \notin \CL(G)$ for all $f\in \CL(G)$. Consider $\ba = (h(P_2), h(P_3), \ldots, h(P_7)) = (1,~0,~0,~1,~\omega,~\omega^2)$. Then the generator matrix of $C_\CL(D,\ba G)$ can be obtained by evaluating the functions in $\{1, \left(\frac{X}{Z}\right),\left(\frac{Y+Z}{X}\right)\}$ at $\{P_2, P_3, \ldots, P_7\}$, i.e.,
$$\left(\begin{array}{cccccc}
1 & 0 & 0 & 1 & \omega & \omega^2  \\
\omega & 0 & 0 & \omega^2 & \omega & \omega^2  \\
\omega^2 & 0 & 0 & \omega^2 & 1 & 1  \\
\end{array}\right).$$
Therefore, $(C_\CL(D,G),\ba C_\CL(D,G))$ forms an LCP.
\end{example}

Starting from a self-dual algebraic geometry code $C_{\CL}(D,G)$,  the following result presents a simple way to produce an LCP of codes. The reader could refer to the recent paper of Sok (~\cite{Sok20}), in which a  study on self-dual algebraic geometry codes was presented.
\begin{cor}\label{cor-5.1}
Let $C_{\CL}(D,G)$ be a self dual algebraic geometry code, where $G$ and $D = P_1 + P_2 + \cdots + P_n$ are divisors of a smooth projective curve $\CX$ of genus $g$ such that $2g-2 < \deg(G) < n$ and $\spt(G) \cap \spt(D) = \emptyset$. If $h \in \FF_q(\CX)$ is a function such that 
$(h(P_1), h(P_2), \ldots, h(P_n)) \in \left(\FF_q\setminus\{0,1\}\right)^n$
and $hf \notin  \CL(G)$ for all $f \in \CL(G)$. Then there exists $\ba \in \FF_q^n$ such that $(C_{\CL}(D,G), \ba C_{\CL}(D,G))$ forms an LCP of codes.
Moreover, $C_{\CL}(D,G)$ is equivalent to $\ba C_{\CL}(D,G)$.    
\end{cor}

Finally, starting from two algebraic geometry codes over non-binary finite fields, one of which is an MDS code,  the following result presents, up to some assumptions,  a simple way to produce an LCP of codes. 

\begin{thm}\label{th-5.1}
Let $C_{\CL}(D,G)$ and $C_{\CL}(D,H)$ be two algebraic geometry codes over $\FF_q$ with $q \geq 3$ such that $\ell(G) + \ell(H) = n$.
If $C_{\CL}(D,H)$ is an MDS code, then there exists $\ba = (a_1, a_2, \ldots, a_n)\in \left(\FF_q^*\right)^n$ such that $(\ba C_{\CL}(D,G),C_{\CL}(D,H))$ is an LCP. In addition, if $C_{\CL}(D,G)$ is an MDS, then the dual of $C_{\CL}(D,G)$ is equivalent to $C_{\CL}(D,H)$.
\end{thm}
\begin{proof}
Let $C_{\CL}(D,G):= [n,k]$ and $C_{\CL}(D,H):= [n,n-k]$ be two algebraic geometry codes.\\
If $C_{\CL}(D,G) \cap C_{\CL}(D,H)=\{0\}$, then $(\ba C_{\CL}(D,G),C_{\CL}(D,H))$ is an LCP for $\ba =(1, 1, \ldots, 1)$ as $\ell(G) + \ell(H) = n$.\\
Otherwise $C_{\CL}(D,G)\cap C_{\CL}(D,H) \neq \{0\}$. Let $\dim(C_{\CL}(D,G) \cap C_{\CL}(D,H)) = l$, where $0 < l \leq \min\{k,n-k\}$. Without loss of generality, we consider the generator matrices of $C_{\CL}(D,G)$ and $C_{\CL}(D,H)$ to be of the form
$$\mathcal{G}_1 = \left(\begin{array}{ccc}
I_l & 0 & P_1  \\
0 & I_{k-l} & P_2
\end{array}\right),~~~
\mathcal{G}_2 = \left(\begin{array}{ccc}
B_1 & B_2 & B_3
\end{array}\right)$$ 
respectively, where $P_1$ is an $l \times (n-k)$ matrix, $P_2$ is an $(k-l) \times (n-k)$ matrix, $B_1$ is an $( n-k)\times l$ matrix, $B_2$ is an $ (n-k)\times( k-l)$ matrix, and $B_3$ is an $ (n-k)\times (n-k)$ matrix. Here, 
$\left(\begin{array}{ccc} I_l & 0 & P_1 \end{array}\right)$
is a generator matrix of $C_{\CL}(D,G) \cap C_{\CL}(D,H)$. So, the matrix $\left(\begin{array}{ccc} I_l & 0 & P_1 \end{array}\right)$ is a part of the matrix 
$\left(\begin{array}{ccc} B_1 & B_2 & B_3 \end{array}\right)$.
Therefore,
$$\left(\begin{array}{c} 
\mathcal{G}_1\\
\mathcal{G}_2
\end{array}\right)
= \left(\begin{array}{ccc}
I_l & 0 & P_1  \\
0 & I_{k-l} & P_2  \\
B_1 & B_2 & B_3
\end{array}\right)
 \sim \left(\begin{array}{ccc}
0 & 0 & 0  \\
0 & I_{k-l} & P_2  \\
B_1 & B_2 & B_3
\end{array}\right).$$ 
Then, we have
$\textit{rank}\left(\begin{array}{c} 
\mathcal{G}_1\\
\mathcal{G}_2
\end{array}\right) = n-l$.
Since $C_{\CL}(D,H)$ is MDS code, $\textit{rank}(B_3) = n-k$.

Let us consider
$$\mathcal{G}_{\ba}=\left(\begin{array}{ccc}
I_l & 0 & P_1  \\
0 & I_{k-l} & P_2
\end{array}\right) 
\left(\begin{array}{cccccccc}
\lambda_1 & 0 & 0 & \cdots & 0 & 0 & \cdots & 0  \\
0 & \lambda_2 & 0 & \cdots & 0 & 0 & \cdots & 0  \\
\vdots & \vdots & \vdots & \ddots & \vdots & \vdots & \ddots & \vdots \\
0 & 0 & 0 & \cdots & \lambda_l & 0 & \cdots & 0  \\
0 & 0 & 0 & \cdots & 0 & 1 & \cdots & 0  \\
\vdots & \vdots & \vdots & \ddots & \vdots & \vdots & \ddots & \vdots \\
0 & 0 & 0 & \cdots & 0 & 0 & \cdots & 1
\end{array}\right)
= \left(\begin{array}{ccccc|cc}
\lambda_1 & 0 & 0 & \cdots & 0 &   \\
0 & \lambda_2 & 0 & \cdots & 0 & 0 & P_1  \\
\vdots & \vdots & \vdots & \ddots & \vdots &  & \\
0 & 0 & 0 & \cdots & \lambda_l \\ \hline
0 & 0 & 0 & \cdots & 0 & I_{k-l} & P_2 
\end{array}\right),$$
with $\lambda_i \in \FF_q^*$. Then
$$\left(\begin{array}{c} 
\mathcal{G}_{\ba}\\
\mathcal{G}_2
\end{array}\right) 
= \left(\begin{array}{ccccc|cc}
\lambda_1 & 0 & 0 & \cdots & 0 &   \\
0 & \lambda_2 & 0 & \cdots & 0 & 0 & P_1  \\
\vdots & \vdots & \vdots & \ddots & \vdots &  & \\
0 & 0 & 0 & \cdots & \lambda_l \\ \hline
0 & 0 & 0 & \cdots & 0 & I_{k-l} &P_2  \\
&  & B_1 &  &  & B_2 & B_3
\end{array}\right).$$ 
As $q \geq 3$, there exists $\lambda_i \in \left(\FF^*_q\right)^n$ such that $\lambda_i \neq 1$ is for all $1\leq i\leq h$. We see that 
$$\textit{rank}\left(\begin{array}{c} 
\mathcal{G}_{\ba}\\
\mathcal{G}_2
\end{array}\right) = n.$$
It is easy to check that $\mathcal{G}_{\ba}$ is a generator matrix of $$\ba C_{\CL}(D,G)=\{(\lambda_1 c_1, \lambda_2 c_2,\dots,\lambda_l c_l, c_{l+1},\dots c_n)~|~(c_1,c_2,\dots,c_n)\in C_{\CL}(D,G))\}.$$ This implies that $\ba C_{\CL}(D,G)\cap C_{\CL}{(D,H)}=\{0\}.$ Hence, $(\ba C_{\CL}(D,G),C_{\CL}(D,H))$ is LCP.\\
For the other part, it is well known that the dual MDS code of $\ba C_{\CL}(D,G)$ is an MDS code $\ba^{-1}C_{\Omega}(D,G)$, where $\ba^{-1}=(a_1^{-1},a_2^{-1},\dots,a_n^{-1}).$ Therefore, $(\ba C_{\CL}(D,G))^\perp$ and $C_{\CL}(D,H)$ are equivalent.
\end{proof} 
\begin{example}
 Let $\FF_q$ be the finite field of $q$ elements with $q\geq 3$ and let $\mathcal{P}_k$ be the vector space of polynomials $f\in \FF_q[\CX]$ with the degree of those polynomials at most $k-1$. $\FF_q(\CX)$ be the corresponding function field with the point of infinity $\mathcal{O}$. For $\textbf{b} = (\alpha_1,\alpha_2,\ldots,\alpha_n)$ where $\alpha_i \in \FF_q, 1 \leq i \leq n$ and pairwise distinct, let $RS_k(\textbf{b})=\{(f(\alpha_1),f(\alpha_2),\ldots,f(\alpha_n))~|~f\in\mathcal{P}_k\}$ be a $k$-dimensional Reed-Solomon codes. Suppose $h(\CX)=\prod_{i=1}^n(\CX-\alpha_i)$ and $h'(\CX)$ is the derivative of $h$ with respect to $\CX$. Now, we associate $RS_k(\textbf{b})$ with the algebraic geometry code as 
 $$RS_k(\textbf{b})=C_{\CL}(D,(k-1)\mathcal{O})~~\text{and}~~RS_k(\textbf{b})^\perp=C_{\CL}(D,(h')+(n-k-1)\mathcal{O}),$$ 
 where $D = \sum\limits_{i=1}^n P_i$, with $P_i = P_{\CX-\alpha_i}$ be the rational places corresponding to the irreducible polynomials $\CX - \alpha_i$ for $1\leq i\leq n$. Note that $RS_k(\textbf{b})$ is an $[n,k,n-k+1]$  an MDS code. Then $RS_k(\textbf{b})^\perp$ is also an $[n,n-k,k+1]$ MDS code as well. By Theorem \ref{th-5.1}, there exists $\ba \in (\FF_q^*)^n$ such that the pair $(\ba RS_k(\textbf{b}), RS_k(\textbf{b})^\perp)$ forms an LCP.
\end{example}
\section{Conclusion}\label{sec:con}

Although current standard cryptographic algorithms are proven to withstand so-called logical attacks (i.e. classical cryptanalyses),
their hardware and software implementations have exhibited vulnerabilities to side-channel attacks (SCA) and fault injection attacks (FIA).  Specifically, it has been observed that the countermeasure against FIA in the DSM scheme could lead to a weakness in security against SCA in some specific environments. This led to a variant of the LCP problem on the codes side, where one has to lengthen the two codes used in DSM while preserving the parameters of the original pair as much as possible.  This cryptographic motivation emphasizes some main coding problems, particularly on  LCP codes. In contrast to LCD, the LCP of codes have been studied intensively less.  In particular, a significant gap was needed regarding the LCP of codes, especially from the explicit designs over finite fields (specifically small fields). Inspired by ~\cite{MTQ18}, we have studied the LCP of algebraic geometry codes in several directions in this paper. Our study included:
\begin{itemize}
\item explicit construction  methods over algebraic curves;
\item examination of the subfamily of cyclic codes;
\item determination of their security parameters;
\item investigation their optimality aspects;
\item investigation the duality aspects.
\end{itemize}
 A MAGMA program has checked all our computational results.\\

For future work, it could be interesting to investigate more  LCP of algebraic geometry codes over other curves and more families of codes, such as the well-known Reed-Muller codes.



\begin{thebibliography}{00}

\bibitem{BL06}
Stephane Ballet and Dominique Le Brigand.
\newblock On the existence of non-special divisors of degree g and g-1 in algebraic function fields over {\(\FF\)}\({}_{\mbox{q}}\).
\newblock {\em Journal of Number Theory}, 116(2):293--310, 2006.

\bibitem{BCC14}
Julien Bringer, Claude Carlet, Herv{\'{e}} Chabanne, Sylvain Guilley, and Houssem Maghrebi.
\newblock Orthogonal direct sum masking - {A} smartcard friendly computation paradigm in a code, with built-in protection against side-channel and fault attacks.
\newblock In {\em Information Security Theory and Practice. Securing the Internet of Things {WISTP} 2014, Crete, Greece. Proceedings}, volume 8501 of {\em Lecture Notes in Computer Science}, pages 40--56. Springer, 2014.

\bibitem{CG16}
Claude Carlet and Sylvain Guilley.
\newblock Complementary dual codes for counter-measures to side-channel attacks.
\newblock {\em Advances in Mathematics of Communications}, 10(1):131--150, 2016.

\bibitem{CG18}
Claude Carlet, Cem G{\"{u}}neri, Ferruh {\"{O}}zbudak, Buket {\"{O}}zkaya, and Patrick Sol{\'{e}}.
\newblock On linear complementary pairs of codes.
\newblock {\em {IEEE} Trans. Inf. Theory}, 64(10):6583--6589, 2018.

\bibitem{CMT18}
Claude Carlet, Sihem Mesnager, Chunming Tang, Yanfeng Qi, and Ruud Pellikaan.
\newblock Linear codes over $\FF_q$ are equivalent to {LCD} codes for $q>3$.
\newblock {\em {IEEE} Trans. Inf. Theory}, 64(4):3010--3017, 2018.

\bibitem{CMTQ18}
Claude Carlet, Sihem Mesnager, Chunming Tang, and Yanfeng Qi.
\newblock Euclidean and hermitian {LCD} {MDS} codes.
\newblock {\em Designs, Codes and Cryptography}, 86(11):2605--2618, 2018.

\bibitem{CMTQ19}
Claude Carlet, Sihem Mesnager, Chunming Tang, and Yanfeng Qi.
\newblock On $\sigma$-LCD codes.
\newblock {\em {IEEE} Trans. Inf. Theory}, 65(3):1694--1704, 2019.

\bibitem{GOS18}
Cem G{\"{u}}neri, Buket {\"{O}}zkaya, and Selcen Sayici.
\newblock On linear complementary pair of \emph{n} {D} cyclic codes.
\newblock {\em {IEEE} Commun. Lett.}, 22(12):2404--2406, 2018.

\bibitem{Mas92}
James~L. Massey.
\newblock Linear codes with complementary duals.
\newblock {\em Discrete Mathematics}, 106-107:337--342, 1992.

\bibitem{JX-LCD} Lingfei Jin and Chaoping Xing.  
\newblock Algebraic Geometry Codes with Complementary Duals Exceed the Asymptotic Gilbert-Varshamov bound. 
\newblock{\em {IEEE} Trans. Inf. Theory}, 64(9): 6277 - 6282, 2018.


\bibitem{MTQ18}
Sihem Mesnager, Chunming Tang, and Yanfeng Qi.
\newblock Complementary dual algebraic geometry codes.
\newblock {\em {IEEE} Trans. Inf. Theory}, 64(4):2390--2397, 2018.

\bibitem{NBD15}
Xuan~Thuy Ngo, Shivam Bhasin, Jean{-}Luc Danger, Sylvain Guilley, and Zakaria Najm.
\newblock Linear complementary dual code improvement to strengthen encoded circuit against hardware trojan horses.
\newblock In {\em {IEEE} International Symposium on Hardware Oriented Security and Trust, {HOST} 2015, Washington, DC, USA, 5-7 May 2015}, pages 82--87. {IEEE} Computer Society, 2015.

\bibitem{Sil85}
Joseph~H. Silverman.
\newblock {\em The Arithmetic of Elliptic Curves}, volume 106 of {\em Graduate texts in mathematics}.
\newblock Springer, isbn 978-3-540-96203-8, 1986.

\bibitem{Sok20}
Lin Sok.
\newblock Explicit constructions of {MDS} self-dual codes.
\newblock {\em {IEEE} Trans. Inf. Theory}, 66(6):3603--3615, 2020.

\bibitem{Sti08}
Henning Stichtenoth.
\newblock {\em Algebraic Function Fields and Codes}.
\newblock Springer Publishing Company, Incorporated, 2nd edition, isbn 978-3-540-76877-7, 2008.

\bibitem{YM94}
Xiang Yang and James~L. Massey.
\newblock The condition for a cyclic code to have a complementary dual.
\newblock {\em Discrete Mathematics}, 126(1-3):391--393, 1994.
\end{thebibliography}



\end{document}